\documentclass[copyright,creativecommons]{eptcs}
\providecommand{\event}{GandALF 2014} 
\usepackage{breakurl}             
\usepackage{xspace}
\usepackage{amsmath,amssymb,amsthm}
\usepackage{latexsym}
\usepackage{algorithm,algpseudocode}
\usepackage{tikz}
\usetikzlibrary{shapes.geometric}
\usepackage{multirow}

\theoremstyle{plain}
\newtheorem{theorem}{Theorem}
\newtheorem{proposition}[theorem]{Proposition}

\theoremstyle{definition}
\newtheorem{definition}[theorem]{Definition}

\theoremstyle{plain}
\newtheorem{example}[theorem]{Example}

\newcommand{\mucalc}{\ensuremath{\mathcal{L}_\mu}\xspace}
\newcommand{\aplayer}{\ensuremath{\mathcal{P}}\xspace}
\newcommand{\odd}{\ensuremath{\mathsf{Odd}}\xspace}
\newcommand{\even}{\ensuremath{\mathsf{Even}}\xspace}
\newcommand{\opponent}{\ensuremath{\overline{\mathcal{P}}}\xspace}

\newcommand{\Prop}{\ensuremath{\mathcal{P}}}

\newcommand{\Var}{\ensuremath{\mathcal{V}}}
\newcommand{\sem}[3]{\ensuremath{{[\![#1]\!]}_{#2}^{#3}}}
\newcommand{\prio}[1]{\ensuremath{\mathsf{prio}_{#1}}}
\newcommand{\own}[1]{#1}

\newcommand{\Nat}{\ensuremath{\mathbb{N}}}

\title{The Fixpoint-Iteration Algorithm for Parity Games}
\author{Florian Bruse \qquad Michael Falk \qquad Martin Lange
\institute{School of Electrical Engineering and Computer Science, University of Kassel, Germany}
}
\def\titlerunning{The Fixpoint-Iteration Algorithm for Parity Games}
\def\authorrunning{F.~Bruse, M.~Falk \& M.~Lange}

\begin{document}
\maketitle

\begin{abstract}
It is known that the model checking problem for the modal $\mu$-calculus reduces to the problem of solving
a parity game and vice-versa. The latter is realised by the Walukiewicz formulas which are satisfied by a node
in a parity game iff player $0$ wins the game from this node. Thus, they define her winning region, and any 
model checking algorithm for the modal $\mu$-calculus, suitably specialised to the Walukiewicz formulas, yields 
an algorithm for solving parity games. In this paper we study the effect of employing the most straight-forward
$\mu$-calculus model checking algorithm: fixpoint iteration. This is also one of the few algorithms, if not the 
only one, that were not originally devised for parity game solving already. While an empirical study quickly 
shows that this does not yield an algorithm that works well in practice, it is interesting from a theoretical
point for two reasons: first, it is exponential on virtually all families of games that were designed as
lower bounds for very particular algorithms suggesting that fixpoint iteration is connected to all those. Second, 
fixpoint iteration does not compute positional winning strategies. Note that the Walukiewicz formulas only define
winning regions; some additional work is needed in order to make this algorithm compute winning strategies. We 
show that these are particular exponential-space strategies which we call eventually-positional, and we show how
positional ones can be extracted from them.  
\end{abstract}

\section{Introduction}

Parity games are 2-person infinite-duration games with important applications in
the area of specification, verification and synthesis of reactive, distributed systems.
They are used as the algorithmic backbone in satisfiability checking for temporal logics 
\cite{journals/corr/FriedmannLL13}, in controller synthesis \cite{AVW-TCS03} and, most 
commonly known, in model checking \cite{Stirling95}. 

Typical problems from these areas reduce to the problem of \emph{solving} a parity game. In the
simplest form this just a \emph{decision} problem: given a parity game, decide which player has 
a winning strategy for it. Due to determinacy of parity games \cite{TCS::Zielonka1998} it is 
always one of the two players. This is a reasonable question in model checking for instance where
parity game solving can be used to decide whether or not a temporal property is satisfied by
a transition system.

However, suppose that the answer is ``no''. Then one usually wants to know the reason for why it
is not satisfied, and the game-theoretic approach to model checking can provide explanations of
this kind easily: the reason for the unsatisfaction is the underlying winning strategy. Equally, 
in satisfiability checking one may not only be interested in whether or not a formula is satisfiable
but may also want to obtain a model in the positive case. Such models can be derived from the underlying
winning strategies. Hence, simply deciding whether or not a winning strategy exists may not be enough
for certain purposes. This becomes very clear when considering controller synthesis: here the task is
not only to decide whether a controller would be synthesisable but to actually generate one from
specifications automatically. Again, this is closely linked to the \emph{computation} of winning
strategies, i.e.\ to solving as a computation problem. 

Parity games enjoy \emph{positional determinacy} \cite{TCS::Zielonka1998}: every game is won by 
either player with a \emph{positional} strategy that selects, for every node controlled by that 
player, one of its outgoing edges. This can be used in a na\"{\i}ve way of turning an algorithm for 
the decision problem into a computation of winning strategies: for every edge check whether or not
removing it changes the outcome of the game. This is obviously impractical, and basically all 
algorithms for solving parity games do in fact compute (positional) winning strategies directly in
all their variants: the recursive algorithm \cite{TCS::Zielonka1998,Schewe/07/Parity}, strategy 
improvement \cite{conf/cav/VogeJ00,conf/csl/Schewe08,FriedmannL:JFCS:2011}, small progress measures 
\cite{Jurdzinski/00,journals/jccs/HKLN}, etc. Also from a purely theoretical point of view it is 
equally desirable to study various ways of computing winning strategies, in particular since the
exact complexity of solving parity games is yet undetermined: it is known to be in NP$\cap$coNP 
\cite{focs91*368} -- for instance by guessing positional strategies and verifying them in 
polynomial time -- and even in UP$\cap$coUP \cite{Jurdzinski/98} but a deterministic polynomial 
time algorithm has not been found yet. The asymptocially best algorithms are subexponential, either
as randomised \cite{journals/tcs/BjorklundV05} or even as a deterministic procedure 
\cite{journals/siamcomp/JurdzinskiPZ08}.

In this paper we study an algorithm for solving parity games that is based on a tight connection to 
model checking for the modal $\mu$-calculus: not only does this model checking problem reduce -- linearly
in both input parameters -- to the problem of solving parity games \cite{Stirling95} and, hence, all 
the aforementioned algorithms can be used for model checking as well. There is also a linear reduction 
in the other direction: a parity game can be seen as a labeled transition system (LTS), and the winning regions
for the players can be defined by in the $\mu$-calculus by the so-called Walukiewicz formulas that, notably,
only depend on the number of different priorities used in the games but not the game structure itself. Thus,
parity games can also be solved -- in the sense of the decision problem -- by applying any $\mu$-calculus 
model checking algorithm to the LTS representing the game and the corresponding Walukiewicz formula. Most
algorithms for model checking the $\mu$-calculus are, however, parity game solving algorithms already with
one exception: computing the semantics of the formula on a finite LTS via explicit fixpoint iteration.
This is, for example, the algorithm underlying BDD-based model checking \cite{IC::BurchCMDH1992}. 

In Section~\ref{sec:prel} we recall parity games and the modal $\mu$-calculus. In Section~\ref{sec:fpiter}
we formulate the fixpoint iteration algorithm for solving parity games based on fixpoint iteration in the
$\mu$-calculus. Note that it only computes winning regions because the Walukiewicz formulas can only define
predicates on game nodes, not game edges. Hence, a natural question concerns the connection to computing
winning strategies. In Section~\ref{sec:winstrat} we provide the answer to this by showing that, unlike
virtually all other known algorithms, fixpoint iteration does not directly compute positional strategies. 
We show that instead it computes finite-memory winning strategies of a special kind that we call 
\emph{eventually-positional} strategies. Positional strategies can be extracted from these at a certain 
expense.

Thus, parity game theory provides positional strategies for model checking, and these are extremely useful
in finding short counterexamples to unsatisfied temporal properties. On the other hand, model checking does
not do the same for parity games. It remains to be understood what exactly causes this imbalance: the
fact that Walukiewicz formulas, fixed up to number of priorities in the game, define winning regions in 
arbitrary parity games; the nature of fixpoint iteration; or possibly even something else.


\section{Preliminaries}
\label{sec:prel}

\subsection{Parity Games}

Parity games are infinite 2-player games played on directed graphs. We call the two players \even and \odd, and write \aplayer for either of them 
and \opponent for their opponent. Formally, a \emph{parity game} is a tuple $\mathcal{G} = (V,V_\even,V_\odd,E,\Omega)$ where $(V,E)$ is a directed graph with
node set $V$ and edge relation $E$ s.t.\ every node has at least one successor. The node set is partitioned into $V_\even$ and $V_{\odd}$ marking those
nodes controlled by player \even, resp.\ \odd. Finally, $\Omega\colon V \to \Nat$ is a function that assigns
to each node a \emph{priority}. The \emph{index} of this game is $|\{\Omega(v) \mid v \in V \}|$, i.e.\ the number of different priorities occurring 
in this game. Since we only consider finite parity games, the index is well-defined.

A \emph{play} is an infinite sequence $v_0,v_1,\ldots$ such that $(v_i,v_{i+1}) \in E$ for all $i \in \Nat$. It is said to \emph{start} in $v_0$. 
The \emph{winner} of this play is determined by $\lim\sup_{i \to \infty} \Omega(v_i)$: it is won by \even if this value, i.e.\ the maximal priority
occurring infinitely often, is even. Otherwise it is won by \odd.

A \emph{strategy} for a player \aplayer is a function $\sigma \colon V^*V_{\aplayer} \to V$ such that $(v,f(uv)) \in E$ for all $v \in V_{\aplayer}$ and 
all $u \in V^*$. Thus, a strategy maps a node owned by player \aplayer to one of its successor nodes. Intuitively, it describes where player 
\aplayer should move to when a play has reached the node $v$ and so far it has gone through the sequence of nodes $u$. A play $(v_i)_{i \in \mathbb{N}}$ 
\emph{adheres} to a strategy $\sigma$ if for all $(v_i)_{i \leq k}$ such that $v_k \in
V_{\aplayer}$ we have $v_{k+1} = f((v_i)_{i \leq k})$. A strategy $\sigma$ is \emph{winning} for player \aplayer in node $v$ if every play that starts in $v$
and adheres to $\sigma$ it is won by \aplayer. We say that \aplayer \emph{wins} $\mathcal{G}$ \emph{in} $v$, or simply \emph{wins} $v$ if $\mathcal{G}$ 
is clear from the context, if \aplayer has a winning strategy in node $v$.  
The \emph{winning region} for \aplayer, usually denoted $W_\aplayer$, in $\mathcal{G}$ is the set 
of all nodes $v$ such that \aplayer wins $\mathcal{G}$ in $v$. The following result states that $W_\even$ and $W_\odd$ partition the node set of a game.

\begin{proposition}[\cite{martin75,focs91*368,TCS::Zielonka1998}]
For every game $\mathcal{G}$ and node $v$ in it either \even or \odd wins $\mathcal{G}$ in $v$.
\end{proposition}

For many applications it is not only interesting to know which player wins a node but also how it is won, i.e.\ what the winning strategy is. For these
purposes, we need to consider special strategies that can be represented finitely. The simplest of these are positional ones, also known as history-free 
or memory-less ones. A strategy $\sigma$ is called \emph{positional} if $\sigma(uv) = \sigma(u'v)$ for all $u,u' \in V^*$. Thus, the decisions made by 
such strategies only depend on the current position in a play, not the history of the play. Such strategies can easily be represented as a subset of
the edge relation $E$. Luckily, when asking for winning strategies it suffices to consider positional ones.

\begin{proposition}[\cite{TCS::Zielonka1998}]
\label{prop:poswinstr}
\aplayer wins  node $v$ in a game $\mathcal{G}$  if and only if \aplayer wins this node with a positional winning strategy.
\end{proposition}

Finally, when solving a game it is not necessary to compute winning strategies for every node independently. Instead, it suffices to compute at most
one per player.

\begin{proposition}[\cite{TCS::Zielonka1998}]
\label{thm:positional-strat-obda}
Let $U$ be a set of nodes in a game $\mathcal{G}$. If \aplayer has (positional) winning strategies $\sigma_u$ for all $v \in U$ that are winning in 
each $v$ respectively, then there is a (positional) winning strategy $\sigma$ for him/her that is winning in all $v \in U$. 
\end{proposition}

We can therefore formulate the problem of \emph{solving} a finite parity game $\mathcal{G}$ as follows: compute the winning regions $W_\even$ and
$W_\odd$ together with corresponding (positional) winning strategies $\sigma_\even$ and $\sigma_\odd$. Clearly, it suffices to compute $W_\even$
\emph{or} $W_\odd$ since they are complements to each other. However, winning strategies for one player cannot easily be obtained from winning
strategies for the other player.


\subsection{The Modal $\mu$-Calculus}

The modal $\mu$-calculus (\mucalc) is usually interpreted over labeled transition systems. Here we are only interested in the application of a particular
$\mu$-calculus model checking algorithm for solving parity games. To this end, one can regard parity games very naturally as labeled transition systems.
However, we avoid introducing them in full generality and simply interpret the modal $\mu$-calculus over parity games directly. Besides saving
some space, this has two simple consequences for the syntax and semantics of the modal $\mu$-calculus presented here. Syntactically, formulas do
not use arbitrary atomic propositions but only $\Prop := \{\own{\even},\own{\odd}\} \cup \{ \prio{i} \mid i \in \Nat \}$ in order to mark the ownership and
priority of a node in a game. Semantically, every node satisfies exactly two of these since every node has a unique owner and unique priority. This
also means that we do not need negation in the logic's syntax.

Let $\Var$ be a set of second-order variables. The syntax of \mucalc-formulas in positive normal form is generated by the following grammar:
\[
\varphi := \own{\aplayer} \mid \prio{i} \mid \overline{\prio{i}} \mid X \mid \varphi \wedge \varphi \mid \varphi \vee \varphi \mid \Diamond \varphi \mid \Box \varphi \mid
\mu X.\varphi \mid \nu X. \varphi
\]
where $\aplayer \in \{\even,\odd\}$, $i \in \Nat$ and $X \in \Var$. 

Given a parity game $\mathcal{G} = (V,V_\even,V_\odd,E,\Omega)$, an \emph{assignment} is a partial mapping $\rho\colon\Var\to 2^{V}$. The semantics of 
a formula $\varphi$ in $\mathcal{G}$ is a set of nodes, inductively defined as follows.
\[
\begin{aligned}
\sem{\own{\aplayer}}{\rho}{\mathcal{G}} \enspace &:=\enspace V_\aplayer \enspace, \text{ for } \aplayer \in \{\even,\odd\} \\
\sem{\prio{i}}{\rho}{\mathcal{G}} \enspace &:= \enspace \{v \in V \mid \Omega(v) = i \} \\
\sem{\overline{\prio{i}}}{\rho}{\mathcal{G}} \enspace &:= \enspace \{v \in V \mid \Omega(v) \ne i \} \\
\sem{X}{\rho}{\mathcal{G}} \enspace &:= \enspace \rho(X) \\ 
\sem{\varphi \vee \psi}{\rho}{\mathcal{G}} \enspace &:= \enspace \sem{\varphi}{\rho}{\mathcal{G}} \cup \sem{\psi}{\rho}{\mathcal{G}} \\
\sem{\varphi \wedge \psi}{\rho}{\mathcal{G}} \enspace &:= \enspace \sem{\varphi}{\rho}{\mathcal{G}} \cap \sem{\psi}{\rho}{\mathcal{G}} \\
\sem{\Diamond \varphi}{\rho}{\mathcal{G}} \enspace &:= \enspace \{v \in V\mid \exists u \in \sem{\varphi}{\rho}{\mathcal{G}}\text{ s.t. } (v,u) \in E\} \\
\sem{\Box \varphi}{\rho}{\mathcal{G}} \enspace &:= \enspace     \{v \in V\mid \forall u \in V: \text{ if }(v,u) \in E \text{ then } u \in \sem{\varphi}{\rho}{\mathcal{G}}\} \\
\sem{\mu X. \varphi}{\rho}{\mathcal{G}} \enspace &:= \enspace    \bigcap\{T \subseteq V \mid \sem{\varphi}{\rho[X\mapsto T]}{\mathcal{G}} \subseteq T   \}  \\
\sem{\nu X. \varphi}{\rho}{\mathcal{G}} \enspace &:= \enspace   \bigcup\{T \subseteq V \mid T \subseteq \sem{\varphi}{\rho[X\mapsto T]}{\mathcal{G}}    \}
\end{aligned}
\]

A simple model checking algorithm for \mucalc can simply compute the semantics of a given formula on a given parity game by induction on the formula
structure, using the definition of extremal fixpoints as meets and joins in a complete lattice. Computing the semantics of fixpoint formulas can be done more efficiently using \emph{fixpoint iteration}. 
%
%
Consider a least fixpoint formula $\mu X. \varphi$, an assignment $\rho$ and a finite parity game $\mathcal{G}$ of $n$ nodes. Define a chain of sets 
$(X_\rho^i)_{i \leq n}$, via 
\[
X_\rho^0 \enspace :=\enspace \emptyset \quad, \qquad 
X_\rho^{i+1} \enspace :=\enspace \sem{\varphi}{\rho[X\mapsto X_\rho^i]}{\mathcal{G}}
\]
It is a standard exercise to show that $\varphi$ is monotone in the variable $X$. Hence, we have $X_\rho^0 \subseteq X_\rho^1 \subseteq \ldots \subseteq X_\rho^n$ and
therefore either $X_\rho^n = V$ or there is some $i < n$ with $X_\rho^i = X_\rho^{i+1}$, i.e.\ this sequence stabilises after at most $n$ iterations. We write $X_\rho^*$
for the value in this sequence when it becomes stable. Note that it equals $X_\rho^n$ but it may of course be obtained much earlier than after $n$ steps,
and this is why we prefer to denote it $X_\rho^*$.

Dually, for a greatest fixpoint formula $\nu Y.\psi$ we define
\[
Y_\rho^0 \enspace :=\enspace V \quad, \qquad 
Y_\rho^{i+1} \enspace :=\enspace \sem{\psi}{\rho[Y\mapsto Y_\rho^i]}{\mathcal{G}}
\]
and obtain $Y_\rho^0 \supseteq Y_\rho^1 \supseteq \ldots \supseteq Y_\rho^n$ with the same stabilisation property and a stable value $Y_\rho^*$. It is well-known 
that the points at which these chains become stable coincide with the corresponding fixpoints.

\begin{proposition}[\cite{Tars55}]
Let $\mu X.\varphi$ and $\nu Y.\psi$ be formulas, let $\rho$ be an assignment, and let $X_\rho^*$, $Y_\rho^*$ be defined as above for some finite parity game $\mathcal{G}$. Then we have
$\sem{\mu X. \varphi}{\rho}{\mathcal{G}} = X_\rho^*$ and $\sem{\nu Y. \psi}{\rho}{\mathcal{G}} = Y_\rho^*$. 
\end{proposition}


Note that these fixpoints may depend on  $\rho$. When using fixpoint iteration for the computation of the value of formulas with 
nested fixpoints like $\nu Y.\mu X.(\prio{17} \wedge \Diamond Y) \vee \Diamond X$ for instance it is necessary to nest the corresponding iterations 
as well. Starting with $Y^0 = V$ and $X^0 = \emptyset$ we obtain $\rho_0\colon Y\mapsto V$. Hence, we can compute $X_{\rho_0}^1, X_{\rho_0}^2, \ldots, X_{\rho_0}^*$ which then becomes $Y^1$. Then it is necessary to 
redo the inner iteration with $\rho_1\colon Y\mapsto Y^1$, instead of $Y \mapsto Y^0$, which may yield a different fixpoint $X_{\rho_1}^*$ which becomes the new value $Y^2$, and
so on. Thus, if there are $n$ nodes in the parity game, this nested fixpoint iteration may take up $n^2$ many iterations: $n$ for the outer, and in each of them an inner one of at most $n$.

In the above explanation of these \emph{approximants} we made the dependency on the assignment visible and wrote $X^*_\rho$ for
instance. As said above, we are only interested in very particular formulas of \mucalc for the purpose of designing a parity game solver, and we  therefore make the assignment visible in a different way in the presentation of this algorithm in the next section.


\section{The Fixpoint Iteration Algorithm}
\label{sec:fpiter}

\subsection{Defining Winning Regions in the Modal $\mu$-Calculus}

Let $\mathcal{G} = (V,V_\even,V_\odd,E,\Omega)$ be a fixed finite parity game. We assume that $\Omega: V \to \{0,\ldots,d-1\}$ for some $d \ge 1$, 
i.e.\ the all priorities occurring in $\mathcal{G}$ are between $0$ and $d-1$ inclusively. Thus, its index is $d$. It is well-known that every parity 
game can equivalenty be transformed into one in which the least priority is either $0$ or $1$, and whenever a larger priority $i$ occurs in the game 
then so does $i-1$; the transformation is also known as \emph{priority compression}. Thus, the priority function can be assumed to map into 
$\{0,\ldots,d-1\}$ or $\{1,\ldots,d\}$. The second case is subsumed by the first, but it is also very simple to adjust what follows to deal with
this range.  

Walukiewicz \cite{DBLP:conf/stacs/Walukiewicz96} has shown that the winning region for player $\even$ in $\mathcal{G}$ can be defined by the
\mucalc formula
\[
\Xi_d \enspace := \enspace \sigma X_{d-1} \ldots \mu X_1.\nu X_0. \big((\own{\even} \to \Diamond(\bigwedge\limits_{i=0}^{d-1} \prio{i} \to X_i)) \wedge 
(\own{\odd} \to \Box( \bigwedge_{i = 0}^{d-1} \prio{i} \to X_i))\big)
\]
Here, $\sigma$ is $\nu$ if $d$ is odd and $\mu$ if it is even.

Note that this is not representable in the syntax defined in the previous section, simply because of the absence of logical implication. It is
possible, though, to rewrite this formula using the fact that the propositions $\own{\even}$ and $\own{\odd}$ are mutually exclusive (every node
satisfies exactly one of these) and that the sets of priorities partition the set of nodes. Moreover, a diamond operator commutes with disjunctions, and a box operator commutes with conjunctions. Thus,
the Walukiewicz formulas are equivalent -- over parity games -- to the following.
\[
\Phi_d \enspace := \enspace \sigma X_{d-1} \ldots \mu X_1.\nu X_0. \big(\underbrace{(\own{\even} \wedge \bigvee\limits_{i=0}^{d-1} \Diamond(\prio{i} \wedge X_i)) \vee 
(\own{\odd} \wedge \bigwedge_{i = 0}^{d-1} \Box(\overline{\prio{i}} \vee X_i))}_{\Psi(X_{d-1},\ldots,X_0)}\big)
\]
The main advantage of this form is, however, not the representability in the syntax of the previous section but an efficiency gain for model checking
and therefore for computing winning regions. Remember the remark above about nested fixpoint iterations. Model checking this formula by fixpoint 
iteration starts by initialising all $X_{d-1},\ldots,X_0$ with $V$ or $\emptyset$ alternatingly, then evaluating $\Psi(X_{d-1},\ldots,X_0)$ which 
yields the value for $X_0^1$. This value now replaces the old value for $X_0$, and a re-evaluation of $\Psi$ under this variable assignment yields
$X_0^2$ etc.\ until $X_0^*$ is found which also becomes $X_1^1$. Then the first iteration for $X_1$ is finished and a new inner iteration with 
$X_0 := V$ needs to be started. 

Note how the evaluation of $\Phi_{d-1}$ by fixpoint iteration continuously evaluates $\Psi$ under assignments to $X_{d-1},\ldots,X_0$ that vary in a very
regular way: the value for some $X_i$, $i < d$ only changes (from $X_i^j$ to $X_i^{j+1}$) when the fixpoint $X_{i-1}^*$ has been found. Thus, many
values for variables $X_i$ stay the same during successive evaluations of $\Psi$, and $\Phi_d$ is designed to make use of that. Whenever the value
of an $X_i$ changes, we only need to recompute the values of $\Diamond(\prio{i} \wedge X_i)$ and $\Box(\overline{\prio{i}} \vee X_i)$ while the
other disjuncts and conjuncts remain the same.

\subsection{Solving Parity Games by Fixpoint Iteration}

\begin{algorithm}[t]
\begin{minipage}{9cm}
\begin{algorithmic}[1]
\Procedure{FPIter}{$\mathcal{G} = (V,V_\even,V_\odd,E,\Omega)$}
  \For{$i \gets d-1,\ldots,0$} \label{line:startinit}
     \State \textsc{Init}($i$)
  \EndFor \label{line:stopinit}
  \Repeat \label{line:loopstart}
    \State $\mathit{count}[0] \gets \mathit{count}[0]+1$
    \State $X'_0 \gets X_0$; $X_0 \gets \textsc{Diamond}() \cup \textsc{Box}()$ \label{line:evalpsi}
    \State $i \gets 0$
    \While{$X_i = X'_i$ and $i < d-1$} \label{line:shiftstart} 
      \State $i \gets i+1$
      \State $\mathit{count}[i] \gets \mathit{count}[i]+1$
      \State $X'_i \gets X_i$; $X_i \gets X_{i-1}$
      \State \textsc{Init}($i-1$) \label{line:initcall}
    \EndWhile \label{line:shiftend}
  \Until{$i = d-1$ and $X_{d-1} = X'_{d-1}$} \label{line:loopend}
  \State \textbf{return} $X_{d-1}$
\EndProcedure
\end{algorithmic}
\end{minipage}
\begin{minipage}{6cm}
\begin{algorithmic}[1]
\Function{Init}{$i$}
  \State $X_i \gets$ \textbf{if} $i$ is even \textbf{then} $V$ \textbf{else} $\emptyset$
  \State $\mathit{count}[i] \gets 0$
\EndFunction
\end{algorithmic}
\vskip3mm
\begin{algorithmic}[1]
\Function{Diamond}{\,\!}
  \State \textbf{return} $\{ v \in V_\even \mid \exists t \in V. (v,t) \in E \text{ and } t \in X_{\Omega(t)} \}$
\EndFunction
\end{algorithmic}
\vskip3mm
\begin{algorithmic}[1]
\Function{Box}{\,\!}
  \State \textbf{return} $\{ v \in V_\odd \mid \forall t \in V. (v,t) \in E \text{ implies } t \in X_{\Omega(t)} \}$
\EndFunction
\end{algorithmic}
\end{minipage}

\caption{Solving parity games by fixpoint iteration.}
\label{alg:fpiter}
\end{algorithm}
 
Algorithm~\ref{alg:fpiter} takes a finite parity game and returns the winning region for player \even in this game. It does so by computing the semantics
of variant $\Phi_d$ of the Walukiewicz formula by fixpoint iteration, assuming that $d-1$ is the maximal priority occurring in $\mathcal{G}$. Function
\textsc{Init} initialises the fixpoint variables with values depending on whether they are being used for a least or a greatest fixpoint iteration.

The main part is the loop in lines \ref{line:loopstart}--\ref{line:loopend}. Here, we evaluate the inner expression $\Psi(X_d,\ldots,X_0)$ with respect
to the current values of these variables. 
The result of the evaluation of the expression $\Psi$ is then stored in
$X_0$ and, whenever the fixpoint for $X_i$ has been reached, also in $X_{i+1}$. This is done in the loop in lines
\ref{line:shiftstart}--\ref{line:shiftend} which then also reset the values of $X_j$ for $j < i$. The variables $X'_i$ are being used to check whether or
not the value for some $X_i$ has changed in an iteration. The procedure terminates when the fixpoint for the outermost iteration has been found. The
result is the last evaluation of $\Psi$ which is then stored in $X_{d-1}$.

Additionally, Algorithm~\ref{alg:fpiter} maintains an array $\mathit{count}$ which stores in its $i$-th entry the current number of the iteration for
variable $X_i$. It has no effect on the returned result but it allows the running time to be estimated elegantly and it is indispensable for an 
extension that also computes winning strategies, to be studied in the next section.

Correctness of this algorithm is a simple consequence of the fact that the Walukiewicz formulas $\Xi_d$ correctly define player \even's winning region
\cite{DBLP:conf/stacs/Walukiewicz96}, that $\Phi_d$ is equivalent to $\Xi_d$, that fixpoint iteration correctly computes the 
semantics of an \mucalc formula and that Algorithm~\ref{alg:fpiter} is obtained by specialisation of a general model checking algorithm for \mucalc by 
fixpoint iteration (e.g.\ \cite{IC::BurchCMDH1992}) to $\Phi_d$. We omit a formal proof of these easily checkable facts.
 
\begin{theorem}
Given a finite parity game $\mathcal{G}$ with $n$ nodes, $e$ edges and index $d$, Algorithm~\ref{alg:fpiter} returns the winning region $W_\even$ in this 
game in time $\mathcal{O}(e \cdot n^d)$.
\end{theorem}

The estimation on the running time can be deduced from the observation that the value of the array $\mathit{count}$ grows lexicographically in each
iteration of the loop in lines \ref{line:loopstart}--\ref{line:loopend}. Every entry can be at most $n$ because each fixpoint iteration for a particular
$X_i$ must reach a stable value for $X_i$ after at most $n$ steps. This bounds the number of iterations of this loop by $\mathcal{O}(n^d)$. Each
iteration can be carried out in time $\mathcal{O}(e)$, remembering that $e \ge n \ge d$. The main task in such an iteration is the evaluation of $\Psi$ in
line \ref{line:evalpsi} which takes this time. The other tasks can be done in constant time; the innermost while-loop can do at most $d$ many iterations.

\section{Computing Winning Strategies from Fixpoint Iteration}
\label{sec:winstrat}

This section studies the possibility to extend the fixpoint iteration algorithm such that it does not only return player \even's winning region
(and, by determinacy, also player \odd's) but also corresponding winning strategies for both players. Intuitively, a computation of at least one
winning strategy must be hidden in the algorithm for otherwise the computation of the winning region would be miraculous. In the following we will
show in detail which winning strategies Algorithm~\ref{alg:fpiter} computes -- even for both players -- in order to find the corresponding winning
regions. 

\subsection{Extracting Strategies by Recording the Evaluation of Modal Operators}

An obvious starting point is the examination of line~\ref{line:evalpsi}. Here, a new value for $X_0$ is being computed, and it is not hard to see
that the finally returned value $X_{d-1}$ was at some point also a value of $X_0$. Note that the value of $X_0$ is shifted into $X_1,X_2,\ldots$ for as
long as no change to a previous value of these variables is detected. The expression in line~\ref{line:evalpsi} computes the union of two sets of
nodes:
\begin{itemize}
\item those that belong to player \even and have a successor which has some priority $i$ and belongs to the current value of $X_i$, and
\item those that belong to player \odd and have no successor which has some priority $i$ but does not belong to the current value of $X_i$.
\end{itemize} 
Say there is a node $v \in V_\even \cap \bigcup_{i=0}^{d-1} \Diamond(P_i \cap X_i)$ at some moment. Then we must have $v \in V_\even$ and there must be 
some $u$ such that $(v,u) \in E$ and $u \in X_{\Omega(u)}$, i.e.\ player \even can move from $v$ to $u$ and $u$ must have been discovered already
as good in a sense because it belongs to some $X_j$. In fact, this inductive reasoning only makes sense if $\Omega(u)$ is odd.
Remember that player \even must try to avoid odd priorities. If $\Omega(u)$ was even then $u$ need not have been discovered as ``good'' already.
Instead it is important that it has not been discarded as ``bad'' already. This vague notion of being ``good'' or ``bad'' for player \even is
made precise by the fact that the sets $X_i$ are being increased in a fixpoint iteration when $i$ is odd, and are being decreased when it is even.  

In any way it seems sensible to record the fact that $u$ is a good successor of $v$ by making the edge $(v,u)$ part of a winning strategy for
player \even. Likewise, whenever some node $v \in V_\odd$ does \emph{not} belong to $V_\odd \cap \bigcap_{i=0}^{d-1} \Box(\overline{P_i} \cup X_i)$ then 
it must have some successor $u$ which does not belong to $X_{\Omega(u)}$ at this time. Then it seems like a reasonable idea to record the edge
$(v,u)$ as part of a winning strategy for player \odd. 

Note that $n^d$ can be much larger than $e$. Thus, there must be games in which diverging strategy decisions of this kind are being made 
at different times in the algorithm's run, i.e.\ a node $v \in V_\even$ could have two successors $u$ and $u'$, and the recorded strategy decisions
could alternate between $(v,u)$ and $(v,u')$ several times. This raises the question of how to deal with such findings when previously a different
edge for some player's strategy has been found. It is not hard to show that the two most obvious answers of always taking the first, respectively
last decision, do not yield correct winning strategy computations. 

\begin{example} \rm
\label{ex:whichstrategy}
Consider the following parity game. We depict nodes in $V_\even$ by diamond shapes and nodes in $V_\odd$ by box shapes. In this example, every node
has a unique priority and we can identify nodes with their priorities.
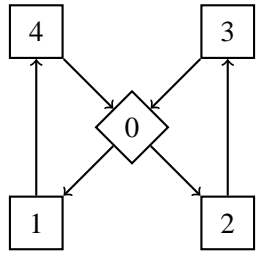
\begin{center}

\begin{tikzpicture}[thick,every node/.style={minimum size=7mm}, node distance=18mm]
  \node[shape=diamond,draw]   (0)                    {$0$};
  \node[shape=rectangle,draw] (1) [below left of=0]  {$1$};
  \node[shape=rectangle,draw] (2) [below right of=0] {$2$};
  \node[shape=rectangle,draw] (3) [above right of=0] {$3$};
  \node[shape=rectangle,draw] (4) [above left of=0]  {$4$};

  \path[draw,->] (0) edge (1)
                     edge (2)
                 (2) edge (3)
                 (3) edge (0)
                 (1) edge (4)
                 (4) edge (0);
\end{tikzpicture}
\end{center}
It is not hard to see that player \even wins from $\{0,1,2,3,4\}$ by choosing the edge $\{(0,1)\}$.

However, executing Algorithm~\ref{alg:fpiter} on this game yields the following findings. In the first iteration,
i.e.\ when $\mathit{count}$ equals $[0,0,0,0,1]$, line~\ref{line:evalpsi} returns $\{0,1,3,4\}$. In particular, node $0$ is included because of the
edge $(0,2)$ since we have $2 \in X_2$ but $1 \not\in X_1$ at this moment. When the fixpoint for $X_1$ has been found and the iteration for $X_0$ is
repeated line~\ref{line:evalpsi} returns $\{0,1,3,4\}$ again, and $\mathit{count}$ has value $[0,0,0,1,1]$. Eventually, node $3$ will enter $X_3$
and at moment $[0,0,1,1,1]$ we compute $\{0,1,3,4\}$ but this time node $0$ is included because of the edge $(0,1)$ and the fact that $1 \in X_1$ whereas
$2 \not\in X_2$. Finally, at moment $[0,1,0,0,1]$ we compute the final winning set $\{0,1,2,3,4\}$ and node $0$ is included in it since $2 \in X_2$
but $1 \not\in X_1$. 

Thus, the following edges are encountered as ``good'' for player \even in this order: $(0,2)$, $(0,1)$, $(0,2)$. Obviously, only $(0,1)$, taken as a 
positional strategy, forms a winning strategy.  
\end{example}

The solution to the question of how the discovered edges in the computation of the diamond- and box-operators in line~\ref{line:evalpsi} can be
assembled to form winning strategies lies in the subtleties of \emph{when} those edges are being discovered. To this end, we amend 
Algorithm~\ref{alg:fpiter} such that it records a strategy decision together with a \emph{timestamp} which is the current value of the
array $\mathit{count}$. We add an array $\mathit{str}$ which holds, for every node $v$, a stack of pairs of timestamps and successor nodes. 
Initially, we assume them to be empty. 

\begin{algorithm}[t]
\begin{minipage}{7.8cm}
\begin{algorithmic}[1]
\Function{Diamond}{\,\!}
  \State $S \gets \emptyset$
  \ForAll{$v \in V_\even$}
     \If{$\exists t \in X_{\Omega(t)}$ with $(v,t) \in E$}
       \State $\mathit{str}[v] \gets \mathit{str}[v] :: (\mathit{count}, t)$
       \State $S \gets S \cup \{ v\}$
     \EndIf
  \EndFor
  \State \textbf{return} $S$
\EndFunction
\end{algorithmic}
\end{minipage}
\begin{minipage}{8cm}
\begin{algorithmic}[1]
\Function{Box}{\,\!}
  \State $S \gets V_\odd$
  \ForAll{$v \in V_\odd$}
     \If{$\exists t \not\in X_{\Omega(t)}$ with $(v,t) \in E$}
       \State $\mathit{str}[v] \gets \mathit{str}[v] :: (\mathit{count}, t)$
       \State $S \gets S \setminus \{ v\}$
     \EndIf
  \EndFor
  \State \textbf{return} $S$
\EndFunction
\end{algorithmic}
\end{minipage}

\caption{Recording strategy decisions in the fixpoint iteration algorithm.}
\label{alg:strat}
\end{algorithm}

Algorithm~\ref{alg:strat} works like Algorithm~\ref{alg:fpiter}; the only difference is the additional recording of strategy decisions in 
the computation of the modal operators through the functions \textsc{Diamond} and \textsc{Box}. Here we only present the amendments of these
functions. The main procedure \textsc{FPIter} remains the same.
Note that the array $\mathit{str}$ contains strategy information for both players: $\mathit{str}[v]$ holds decisions for $\aplayer$ when 
$v \in V_\aplayer$.

\subsection{Eventually-Positional Strategies and a Finite Pay-Off Game}

Our next aim is to prove that Algorithm~\ref{alg:strat} does indeed compute winning strategies for both players. First we need to reveal the nature of these strategies. As said before, they are not positional but they are finite-memory. 

\begin{definition}
Let $\mathcal{G} = (V,V_\even,V_\odd,E,\Omega)$ be a finite parity game and let $\sigma: V^*V_\aplayer \to V$ be a strategy for player \aplayer. We say 
that $\sigma$ is \emph{eventually-positional} if there is a $k \in \Nat$ such that for all $v \in V_\aplayer$ and all $w,w' \in V^*$ with 
$|w| \ge k$ and $|w'| \ge k$ we have $\sigma(wv) = \sigma(w'v)$.
\end{definition}

Intuitively, an eventually-positional strategy forces a player to commit to a positional strategy after a certain amount of time in which the
strategy may depend arbitrarily on the history of a play. In that sense, eventually-positional strategies are finite-memory strategies but
the finite memory is only used at the beginning and can be discarded at some point. Moreover, the positional part of any eventually-positional winning strategy forms a positional winning strategy.

In order to show that Algorithm~\ref{alg:strat} correctly computes eventually-positional winning strategies for parity games we introduce an
auxiliary pay-off game. It is played on the graph of a parity game $(V,V_\even,V_\odd,E,\Omega)$ with $n$ nodes and $d$ priorities as follows. 
Starting in a particular node $v_0$, the two players are given a \emph{credit} which is a $d$-tuple of values between $0$ and $n$. For a credit $C$ and a priority $h$ we write $C(h)$ for the $h$-th value in this tuple. 

A play is a finite sequence of the form $(v_0,C_0),(v_1,C_1),\ldots$ such that $v_0$ is the starting node and $C_0$ is the credit given initially.  A
play is finished in a configuration $(v_k,C_k)$ if $C_k(\Omega(v_k))= 0$. If $\Omega(v_k)$ is even then player \odd has exceeded all his credit and
player \even wins. Otherwise player \odd wins if $\Omega(v_k)$ is odd. In the $i$-th round when $v_i \in V_\aplayer$ player \aplayer chooses $v_{i+1}$ 
among $v_i$'s successors. The new credit value is obtained deterministically as follows.
\[
C_{i+1}(h) \enspace = \enspace \begin{cases}
C_i(h) &, \text{ if } h > \Omega(v_i) \\
C_i(h)-1 &, \text{ if } h = \Omega(v_i) \\
n &, \text{ if } h < \Omega(v_i)
\end{cases}
\]
Thus, whenever the game moves away from a node of priority $h$, some of the $h$-credit must be payed but the players gain the maximal credit for all
values below $h$ again. The credit for values larger than $h$ remains untouched. A player who cannot move loses, but notice that the winning condition is
checked first in each round and only when it does not hold the corresponding player has to move.

These pay-off games form the connection between eventually-positional winning strategies for parity games and Algorithm~\ref{alg:strat}: we will show that this algorithm does indeed compute winning strategies for this pay-off game, and that winning strategies in the pay-off game for a sufficiently large initial credit give rise to winning strategies in parity games.


\begin{theorem}
\label{thm:payoff2parity}
Let $\mathcal{G} = (V,V_\even,V_\odd,E,\Omega)$ be a finite parity game with $n$ nodes and index $d$. Player \aplayer wins the parity game 
$\mathcal{G}$ from any node $v \in V$ with an eventually-positional winning strategy iff he/she wins the pay-off game from $v$ with the initial 
value $(n,\dotsc,n)$.
\end{theorem}
 
\begin{proof}
  Consider the parity game $\mathcal{G}$ from node $v$ and assume that \even has a winning strategy $\sigma$ for the pay-off game with initial credit 
  $C_0 = (n,\dotsc,n)$ from $v$. The case of player \odd is identical. 
  We claim that \even wins the underlying parity game with $\sigma$ as well. Assume that his/her opponent plays with their best strategy 
  against $\sigma$. The result is an infinite play $v_0,v_1,\ldots$, and this can be lifted to a play $(v_0,C_0),(v_1,C_1),\ldots$ in the pay-off game 
  by starting with $C_0$ and adding consecutive credits deterministically. By assumption, this play is winning for \even in the pay-off game. Therefore,
  there must be some $k$ and an even $h$ such that $C_k(h) = 0$. Thus, there are some $j_0,j_1,\ldots,j_n$ with $0 \le j_1 < j_2 < \ldots j_n = k$ such 
  for all $i=0,\ldots,n$ we have $\Omega(v_{j_i}) = h$ and $C_{j_i}(h) = n-i$ and no odd priority larger then $h$ has occurred between $v_{j_0}$ and $v_k$.
  In other words: for some value in the counter at an even position to decrease down to $0$ from $n$ the play must have visited $n+1$ nodes of that priority
  with no higher priority in between. Since there are only $n$ nodes in $\mathcal{G}$ there must be $i,i'$ such that $i < i'$ and $v_{j_i} = v_{j_{i'}}$.
  In particular, this play contains a cycle in which the highest priority is even. In the parity game, \even can play according to $\sigma$ and repeat
  decision when a cycle has occurred with this property. This is obviously winning for her, and it is eventually-positional.

  Conversely, assume that \even has an eventually positional winning strategy for the parity game. According to Proposition~\ref{prop:poswinstr} she also
  has a positional winning strategy. This is also winning for the associated pay-off game: take a play $(v_0,C_0),(v_1,C_1),\ldots,(v_k,C_k)$ of the
  pay-off game in which \even chose successors according to her positional winning strategy in the parity game. Assume otherwise that \odd wins this play
  in the pay-off game. As above, it would have to contain a cycle on which the highest priority is odd. Since we assumed \even to play with a positional
  strategy, this play in the parity game -- continued ad infinitum -- would keep on going through this cycle. Therefore it would be won by \odd which
  contradicts the assumption that \even's strategy was winning.
\end{proof}



\begin{theorem}
\label{thm:payoff2algorithm}
Let $\mathcal{G} = (V,V_\even,V_\odd,E,\Omega)$ be a finite parity game with $n$ nodes and index $d$. Let $v$ be a node of priority $h$. Then \even has a winning strategy with initial credit $C = (c_{d-1},\ldots,c_0)$ from $v$ if and only $v \in X_h$ at moment  $C = (c_{d-1},\ldots,c_0)$ during a run of Algorithm~\ref{alg:strat}.
\end{theorem}
\begin{proof}
  Note that Algorithm~\ref{alg:strat} does not usually produce output for every possible timestamp, since a fixpoint iteration can stabilise well
  before the counter for a specific fixpoint reaches $n$. However, if $h$ is a priority and the value for $X_h$ at timestamp
  $(c_{d-1},\dotsc,c_{h+1},c_h,c_{h-1},\dotsc,c_0)$ equals that of $(c_{d-1},\dotsc,c_{h+1},c_h+1,c_{h-1},\dotsc,c_0)$, then the values for all $X_i$ are
  equal for all timestamps $(c_{d},\dotsc,c_{h+1},c'_{h},c'_{h-1},\dotsc,c'_0)$ such that $c'_i \geq c_i$ for all $i \leq h$. Hence, for the sake of the
  proof we can assume that Algorithm~\ref{alg:strat} keeps on iterating the values of any $X_i$ until the counter reaches $n$ for that $X_i$.

  The proof is by induction. If $C = (0,\dotsc,0)$ then a node $v$ of priority $h$ is in $X_h$ if and only if $h$ is even and player \even wins the
  pay-off game for $v$ only from nodes of even priority.
%
%
  Let $C= (c_{d-1},\dotsc,c_0) \ne (0,\dotsc,0)$ and assume that the claim has been proved for all timestamps that are lexicographically smaller. Assume
  that $v \in X_h$ at this moment. If $c_h = 0$ then $h$ must be even and \even wins right away. Consider $C' = (c_{d-1},\dotsc,c_h
  -1,n,\dotsc,n)$. There are two cases: either $v \in V_{\even}$ and there is a successor $v'$ of priority $h'$ that is in $X_{h'}$ at moment $C'$.  By
  the induction hypothesis, \even has a winning strategy from $v'$ with initial credit $C'$, so moving to $v'$ is a winning strategy for $v$. Or $v \in
  V_{\odd}$ and all successors of $v$ are in their respective sets at moment $C'$. By the induction hypothesis, \even has a winning strategy from these
  successors for the game with initial credit $C'$ and, hence, also wins from $v$ with initial credit $C$.
Conversely, assume that $v \notin X_h$ at moment $C$. If $h = 0$ then $h$ is odd and \odd wins right away. Otherwise, an argument as above shows that \odd has a winning strategy for the pay-off game.
\end{proof}

\subsection{Retrieving Positional Strategies}

Theorems~\ref{thm:payoff2parity} and \ref{thm:payoff2algorithm}, together with the fact that Algorithm~\ref{alg:strat} will compute the winning region
for player \even at moment $(n,\ldots,n)$ -- but may terminate earlier because of monotonicity -- link this algorithm to the computation of
eventually-positional winning strategies for parity games and, hence to the computation of positional strategies. However, these strategies are not
obviously present as Example~\ref{ex:whichstrategy} shows. This is because the stored decisions refer to instances of the pay-off game with different
initial credits. Remember that the correct positional decision for player \even in that game was found as that in the middle of three occurring choices:
The choice discovered first by the algorithm belonged to a strategy for an instance of the pay-off game with initial credit not sufficiently high to lift
it into one for the parity game, and the choice discovered last belongs to the non-positional part of an eventually-positional strategy. It remains to
be seen how the data structure $\mathit{str}$ represents such an eventually-positional strategy for the parity game and, in extension, a positional
strategy.

\begin{algorithm}[t]
\begin{minipage}{8.9cm}
\begin{algorithmic}[1]
\Procedure{ExtractStrategy}{}
  \State $n \gets |V|$
  \ForAll{$v \in V$}
    \State $\sigma(v) \gets \bot$; $\mathit{last}(v) \gets \bot$
  \EndFor
  \While{$\exists v$ with $\sigma(v) = \bot$} 
    \State \textsc{Traverse}($\mathit{Owner}(v),v,(n,\ldots,n)$)
  \EndWhile
\EndProcedure
\end{algorithmic}
\vskip3mm
\begin{algorithmic}[1]
\Function{Strip}{$\aplayer,v,C$}
   \State $l \gets |\mathit{str}(v)|$
   \State delete from $\mathit{str}(v)$ all $(w,C')$ with $C' >_\aplayer C$
   \State \textbf{return} $|\mathit{str}(v)| < l$
\EndFunction
\end{algorithmic}
\end{minipage}
\begin{minipage}{7cm}
\begin{algorithmic}[1]
\Function{Traverse}{$\aplayer,v,C$}
  \If{$v \in V_\aplayer$}
    \If{\textsc{Strip}($\aplayer,v,C$)}
      \State \textbf{let} $(w,C')::\ldots = \mathit{str}(v)$
      \State $\sigma(v) \gets w$; $\mathit{last}(v) \gets C'$
      \State \textsc{Traverse}($\aplayer,w,C'$)
    \EndIf
  \Else
    \ForAll{$u \in V$ with $(v,u) \in E$}
      \State \textsc{Traverse}($\aplayer,u,C$)
    \EndFor
  \EndIf
\EndFunction
\end{algorithmic}
\end{minipage}
\caption{Extracting a positional strategy from the data structures built by Algorithm~\ref{alg:strat}.}
\label{alg:extract}
\end{algorithm}

Indeed we show directly how a positional strategy can be extracted from it. This is done by Algorithm~\ref{alg:extract}. It implements a depth-first
search through the game graph that is guided with a timestamp $C$. Intuitively, we extract from $\mathit{str}$ at some node $v$ the winning strategy
for the pay-off game with credit $C$ or, alternatively, the choice at moment $C$ in the eventually-positional winning strategy for the parity game.
When a node $v$ is being reached with credit $C$ we discard all stored choices for moments that were later than $C$ and continue with the latest
moment $C'$ that was before $C$. Thus, the positional strategy is recovered from the run of Algorithm~\ref{alg:strat} back-to-front: choices
that were discovered late in the algorithm's run are those that are made early in the strategy and do not belong to the positional part of the strategy. This is simply because Algorithm~\ref{alg:strat}
discovers nodes by backwards search through the game graph. 

The question of when some timestamp $C'$ is \emph{later} than another timestamp $C$ needs to be made precise. We distinguish two orders on 
timestamps depending on the player for whom a winning strategy is to be extracted. Let $C' = (c'_{d-1},\ldots,c'_0)$ and $C = (c_{d-1},\ldots,c_0)$. 
Then we have $C' >_\even C$, resp.\ $C' >_\odd C$, if there is some \emph{odd}, resp.\ \emph{even}, $j$ such that $c'_j > c_j$ and $c'_h = c_h$ for all 
\emph{odd}, resp.\ \emph{even}, $h > j$. Thus, when building player \even's strategy we only consider odd indices and vice-versa. The reason is 
the simple fact that iterations for some $X_j$ with an even $j$ are greatest fixpoint iterations that gradually remove nodes from the winning sets.
But the reasons for removal are being discovered in the function \textsc{Box} only which builds a strategy for player \odd. In other words, player
\even's strategy decisions are only being discovered in function \textsc{Diamond} and are only meaningfully recorded for least fixpoint iterations,
i.e.\ for those $X_j$ with an odd index $j$. This is why the even positions in the timestamps bear no information for player \even about which
decision should succeed another in an eventually-positional strategy.   

It remains to argue why this procedure is correct. Clearly, if Algorithm~\ref{alg:extract} returns an array $\mathit{str}$ that has a strategy for every node in $V_{\aplayer}$ in \aplayer's winning region, then these strategy decisions form a positional winning strategy for \aplayer. Moreover, if Algorithm~ \ref{alg:extract} were to remove all strategy decisions for a node, this would amount to constructing a winning strategy in the pay-off game for the opposite player. Since the fixpoint iteration algorithm correctly computes winning strategies for this game, Algorithm~\ref{alg:extract} can never remove all entries for a node.

It is worth pointing out that Algorithm~\ref{alg:extract} does nothing else than searching for a successful progress measure annotation (cf.~\cite{Jurdzinski/00}) within the stack of timestamps: It is not hard to see that any array $\mathit{str}$ returned by the algorithm forms  a successful progress measure annotation, and, conversely, any such progress measure annotation of a parity game allows to convert the annotation at each node into an initial credit such that the player who has the annotation wins the associated payoff game.



\section{Practical Considerations}

\subsection{Implemention and Optimisations}

The fixpoint iteration algorithm has been implemented in \textsc{PGSolver}\footnote{\url{https://github.com/tcsprojects/pgsolver}}, 
an open-source collection of algorithms, tools and benchmarks for parity games \cite{fl-atva09}. The computation of winning regions
and strategies in this implementation works in principle as described above but has been optimised as follows:
\begin{itemize}
\item \textbf{Restriction to nodes of single priorities only}. Note that a variable $X_j$ in the Walukiewicz formulas of the form
  $\Phi_d$ is only ever used in the context of $\prio{i} \wedge X_i$ or $\overline{\prio{i}} \vee X_i$. It is therefore possible to
  restrict the content of the sets $X_i$ in Algorithm~\ref{alg:fpiter} to nodes of priority $i$. This has the advantage of reaching
  fixpoints much earlier in general: a fixpoint iteration for $X_i$ then only needs as many steps as there are different nodes of
  priority $i$. This does not improve the asymptotic worst-case estimation drastically but it makes a difference in practice.

  Note, though, that the implementation of this optimisation requires a slight re-writing of Algorithm~\ref{alg:fpiter}. Currently,
  variable $X_0$ is being used to store the winning sets of the underlying pay-off game, and this value is propagated to 
  $X_1,X_2,\ldots$ for as long as corresponding fixpoints have been reached. With this optimisation we would have to compute the
  current winning set in, say, $X$ and then store in $X_0$ the set $X \cap \Omega^{-1}(0)$, and, once this has become stable, store $X \cap \Omega^{-1}(1)$ in $X_1$, etc. 

\item \textbf{Avoiding costly modal operators}. The most costly operations in the evaluation of the expression $\Psi$ in the 
  Walukiewicz formulas are the modal diamond- and box-operators. In particular the second one requires a traversal of the graph 
  back along one edge and then forth along all edges. It is advisable to keep the number of times that these have to be carried
  out low. This is also the reason for re-writing the original Walukiewcz formulas $\Xi_d$ into their equivalent variants $\Phi_d$.
  Note that in one iteration of the while-loop in Algorithm~\ref{alg:fpiter} only the values of variables $X_0,\ldots,X_m$ for some
  $m$ with $0 \le m \le d$ change whereas the values of $X_{m+1},\ldots,X_{d-1}$ remain the same as in the previous iteration. Using the
  formulation $\Phi_d$ of the Walukiewicz formulas we can store and re-use the value of $\bigvee_{i=m+1}^d \Diamond(\prio{i} \wedge X_i)$
  (and similarly for the box-part) and therefore only need to re-compute the modal operators on the small parts of the graph that
  have changed recently. 

\item \textbf{Elimination of initialisations}. Note that Algorithm~\ref{alg:fpiter} re-initialises all variables $X_0,\ldots,X_{i-1}$
  when an iteration for a variable $X_i$ has been finished. This correctly reflects the semantics of nested fixpoint expressions but
  it is also known that it can be optimised. One can make use of monotonicity of the underlying expression and avoid unnecessary
  initialisations which reduces the overall complexity from exponential in $d$ to exponential in $\lceil d/2\rceil$ 
  \cite{Browne:1997:IAE,seidlfp96}. This optimisation is implemented by first checking the smallest index of a variable $X_i$ for 
  which the current fixpoint has not been reached. The formulation in Algorithm~\ref{alg:fpiter} always resets, whilst finding this index,
  the values of $X_0,\ldots,X_{i-1}$. In an optimised way we would first find the index $i$ and then only reset the values
  of $X_{i-1},X_{i-3},X_{i-5},\ldots$ Note that the values of $X_i$ for odd $i$ always increase monotonically whereas those of the
  $X_i$ for even $i$ always decrease monotonically. Thus, it is only necessary to re-set those with indices of the other kind.
\end{itemize}

\subsection{Benchmarks and Lower Bounds}

\begin{table}[t]
\begin{center}
\begin{tabular}{l|l||r|r|r||r|r|r|r}
 \multicolumn{2}{c||}{} & \multicolumn{3}{c||}{input game} & \multicolumn{4}{|c}{solvers} \\
\multicolumn{2}{l||}{benchmark}          & \#nodes & \#edges & index &      REC &       SPM &        SI &  FPIter \\ \hline\hline
\multirow{4}{*}{elevator}     & $n=3$    &     564 &     950 &     3 &    0.01s &     0.03s &     0.03s &   0.11s \\
                              & $n=4$    &    2688 &    4544 &     3 &    0.05s &     0.20s &     0.59s &   1.37s \\
                              & $n=5$    &   15684 &   26354 &     3 &    0.33s &     2.05s &    37.45s &  25.38s \\
                              & $n=6$    &  108336 &  180898 &     3 &    3.01s &    21.35s & $\dagger$ & 389.27s \\ \hline
\multirow{4}{*}{NBA/DBA inclusion}
                              & $n=50$   &    4928 &   54428 &     3 &    0.08s &     0.37s &   156.39s &   0.67s \\
                              & $n=100$  &    9928 &  208928 &     3 &    0.28s &     1.31s & $\dagger$ &   2.56s \\
                              & $n=150$  &   14928 &  463428 &     3 &    0.61s &     3.45s & $\dagger$ &   6.41s \\
                              & $n=200$  &   19928 &  817928 &     3 &    1.23s &     6.45s & $\dagger$ &  11.20s \\ \hline
\multirow{4}{*}{Hanoi}        & $n=5$    &     972 &    1698 &     2 &    0.00s &     0.12s &     0.18s &   0.14s \\
                              & $n=6$    &    2916 &    5100 &     2 &    0.02s &     0.71s &     2.46s &   1.02s \\
                              & $n=7$    &    8748 &   15306 &     2 &    0.08s &     5.13s &    49.20s &   7.56s \\
                              & $n=8$    &   26244 &   45924 &     2 &    0.24s &    29.71s & $\dagger$ &  50.34s \\ \hline\hline
\multirow{4}{*}{rec.\ ladder} & $n=8$    &      40 &      85 &    27 &    0.00s & $\dagger$ &     0.00s &   0.39s \\
                              & $n=10$   &      50 &     107 &    33 &    0.01s & $\dagger$ &     0.01s &   3.48s \\
                              & $n=12$   &      60 &     129 &    39 &    0.02s & $\dagger$ &     0.01s &  25.45s \\
                              & $n=14$   &      70 &     151 &    45 &    0.06s & $\dagger$ &     0.01s & 265.39s \\ \hline
\multirow{3}{*}{Jurdzinski}   & $n=5$    &      51 &     121 &    13 &    0.00s &     0.00s &     0.00s &   1.98s \\
                              & $n=6$    &      60 &     143 &    15 &    0.00s &     0.00s &     0.00s &  15.27s \\
                              & $n=7$    &      69 &     165 &    17 &    0.00s &     0.01s &     0.00s & 111.66s \\ \hline
\multirow{3}{*}{Friedmann}    & $n=2$    &      44 &     101 &    48 &    0.00s & $\dagger$ &     0.00s &   1.11s \\
                              & $n=3$    &      77 &     186 &    66 &    0.00s & $\dagger$ &     0.04s &  32.99s \\
                              & $n=4$    &     119 &     296 &    84 &    0.00s & $\dagger$ &     0.07s & $\dagger$ \\
\end{tabular}
\end{center}
\caption{An empirical comparison of fixpoint iteration against other algorithms.}
\label{tab:benchmarks}
\end{table}

Table~\ref{tab:benchmarks} contains some runtime results on standard benchmarks included in \textsc{PGSolver}
compared against three other prominent parity game solving algorithms: the recursive one (REC) \cite{TCS::Zielonka1998}, small progress 
measures (SPM) \cite{Jurdzinski/00}, and strategy improvement (SI) \cite{conf/csl/Schewe08}. All optimisations that \textsc{PGSolver} 
provides have been disabled in order to see the algorithms' pure behaviour.

The experiments have been carried out on a MacBook Pro with an Intel Core i5 2.53GHz processor and 8GB of memory. The benchmarks contain two
groups: verification problems and handcrafted tough games. The former contain the verification of an elevator system against a fairness property
(\texttt{elevetorverificatio} $n$), the problem of deciding language inclusion between particular NBA and DBA (\texttt{langincl} $n$ \texttt{10})
and the Towers of Hanoi puzzle as a simple 1-person reachability, i.e.\ parity game of index 2. The latter include the recursive ladder games 
designed for REC (\texttt{recursiveladder} $n$) \cite{DBLP:journals/ita/Friedmann11}, Jurdzinski games of width $3$ for SPM 
(\texttt{jurdzinskigame} $n$ \texttt{3}) \cite{Jurdzinski/00}, Friedmann games for SI (\texttt{stratimprgen -pg friedmannsubexp} $n$)
\cite{conf/lics/Friedmann09}. All six are bundled with \texttt{PGSolver} and described in detail in its manual.\footnote{\url{https://github.com/tcsprojects/pgsolver/blob/master/doc/pgsolver.pdf}}  

Table~\ref{tab:benchmarks} presents the experiments' results in terms of the algorithms' running times and the relevant input parameters.\footnote{Here 
the index is measured as the maximal minus minimal priority plus one.} The 
symbol $\dagger$ marks a time-out after at least 10min. It is noticable that fixpoint iteration can beat strategy improvement, in particular on 
the verification problems. There, it shows similar but worse behaviour compared to the small progress measures algorithm. On the handcrafted
games we can see some unexpected results: it is significantly better than SPM in two cases but in that one case which is designed to show SPM's exponential
behaviour SPM is much better than fixpoint iteration. This shows that the reality of which algorithm performs well or badly on which kinds of games
is not sufficiently understood yet.  

In these six competitions fixpoint iteration only ends up last once; in all other cases it beats either small progress measures or strategy improvement. 
On the other hand, it is nowhere near as good in practice as the recursive algorithm but that seems to be unbeatable in general anyway.  
 
It is noticable that fixpoint iteration is exponential on all three handcrafted benchmark families which were designed to exhibit exponential
runtime of REC, SPM and SI, respectively. This suggests the existence of connections between fixpoint iteration and each of these algorithms.
At least in the case of SPM this is clear: the extraction of an eventually-positional strategy from the strategy decision annotations after the
run of Algorithm~\ref{alg:strat} is nothing more than the search for a small progress measure. Fixpoint iteration precomputes the winning regions 
which SPM does not. However, it is known that searching for a winning strategy is not conceptually easier when the winning regions are known. This
may also explain why fixpoint iteration does not outperform SPM in general. Connections to the other two algorithms are less clear and remain to
be investigated properly in the future.



\bibliographystyle{eptcs}
\bibliography{./literature}

\end{document}